\documentclass[a4paper,11pt]{article}

\title{Well-Separation and Hyperplane Transversals in High Dimensions
\footnote{H. Bergold is supported
	by the German Science Foundation within the research
	training group `Facets of Complexity' (GRK 2434).
	N. Grelier is supported by the Swiss National Science Foundation 
	within the collaborative DACH project 
	\emph{Arrangements and Drawings} as SNSF Project 200021E-171681.
	W. Mulzer is supported in part by ERC StG 757609
	and by the German Research Foundation within the collaborative 
	DACH project \emph{Arrangements
	and Drawings} as DFG Project MU 3501/3-1.
	P. Schnider has received funding from the European Research Council under the European Unions Seventh Framework Programme ERC Grant agreement ERC StG 716424 - CASe.
}
}

\date{ }

\usepackage{amsmath,amssymb,amsfonts}
\usepackage{amsthm}
\usepackage[margin=1in]{geometry}
\usepackage{fge}
\usepackage{mathrsfs}
\usepackage{graphicx}
\usepackage{hyperref}
\usepackage{thm-restate}
\usepackage{color}
\theoremstyle{plain}
\newtheorem{theorem}{Theorem}[section]
\newtheorem{lemma}[theorem]{Lemma}

\newtheorem{claim}[theorem]{Claim}

\theoremstyle{definition}

\theoremstyle{remark}

\usepackage{xspace}
\usepackage{cite}

\newcommand{\abbrev}[2]{\expandafter\newcommand\csname #1\endcsname{#2\xspace}}
\abbrev{Caratheodory}{Carath\'eodory}
\abbrev{Barany}{B\'ar\'any}
\abbrev{Matousek}{Matou{\v{s}}ek}
\abbrev{Lovasz}{Lov{\'a}sz}
\abbrev{Frederic}{Fr\'ed\'eric}

\abbrev{suchthat}{s.t.\@}
\abbrev{wrt}{w.r.t.\@}
\abbrev{ie}{i.e.\@}
\abbrev{etal}{et al.\@}

\newcommand{\R}{\ensuremath{\mathbb{R}}}

\newcommand{\Z}{\ensuremath{\mathbb{Z}}}

\def\inst#1{$^{#1}$}

\begin{document}
	
\author{
	Helena Bergold\inst{1}
	\and
	Daniel Bertschinger\inst{2}
	\and
	Nicolas Grelier\inst{2}
	\and
	Wolfgang Mulzer\inst{1}
	\and
	Patrick Schnider\inst{3}
}	
	
\maketitle

\begin{center}
	{\footnotesize
		\inst{1} 
		Institut f\"ur Informatik, \\
		Freie Universit\"at Berlin, Germany,\\
		\texttt{helena.bergold@fu-berlin.de};\\
		\texttt{mulzer@inf.fu-berlin.de}
		\\\ \\
		\inst{2} 
		Department of Computer Science, \\
		ETH Z\"{u}rich\\
		\texttt{daniel.bertschinger@inf.ethz.ch};\\
		\texttt{nicolas.grelier@inf.ethz.ch}
		\\\ \\
		\inst{3} 
		Department of Mathematical Sciences, \\
		University of Copenhagen \\
		\texttt{ps@math.ku.dk}
		\\\ \\
	}
\end{center}

\begin{abstract}A family of $k$ point sets in $d$ dimensions is 
\emph{well-separated} if the convex hulls of any two disjoint
subfamilies can be separated by a hyperplane. Well-separation is
a strong assumption that allows us to conclude
that certain kinds of generalized
ham-sandwich cuts for the point sets exist. But how hard is it to
check if a given family of high-dimensional point sets
has this property?
Starting from this question, we study several algorithmic
aspects of the existence of transversals and separations in 
high-dimensions.

First, we give an explicit proof that $k$ point sets are
well-separated if and only if their convex hulls admit no 
\emph{$(k - 2)$-transversal},
i.e., if there exists no $(k - 2)$-dimensional flat that intersects
the convex hulls of all $k$ sets.
It follows that  the task of checking 
well-separation lies in the complexity class coNP. 
Next, we show that it is NP-hard to  decide whether there is a 
hyperplane-transversal (that is, a $(d - 1)$-transversal) of a family of 
$d + 1$ line segments in $\R^d$, where $d$ is part of the input. 
As a consequence, it follows that the general problem
of testing well-separation is coNP-complete. 
Furthermore, we show that finding a hyperplane that maximizes the 
number of intersected sets is NP-hard, but allows for 
an $\Omega\left(\frac{\log k}{k \log \log k}\right)$-approximation 
algorithm that is  
polynomial in $d$ and $k$, when each set consists of a single point. 
When all point sets are finite, we show that
checking whether there exists a $(k - 2)$-transversal is 
in fact strongly NP-complete. 

Finally, we take the viewpoint
of parametrized complexity, using the dimension $d$ as a parameter: 
given $k$ convex sets 
in $\R^d$, checking whether there is a $(k-2)$-transversal is FPT 
with respect to $d$. On the other hand, for $k\geq d+1$ 
finite point sets in $\R^d$, it turns out that checking whether there is a 
$(d-1)$-transversal is $W[1]$-hard with respect to $d$.
\end{abstract}

\section{Introduction}
In the study of high-dimensional ham-sandwich cuts,
the following notion has turned out to be fundamental:
we call $k$ sets $S_1, \dots, S_k$ in $\R^d$ are 
\emph{well-separated} if for any \emph{proper} index set $I \subset [k]$ 
(i.e., $I$ is neither empty nor all of $[k]$), the convex hulls of
$S_I = \cup_{i \in I} S_i$ and of $S_{[k] \setminus I} = 
\cup_{i \in [k] \setminus I} S_i$ can be separated by a hyperplane. 
Since any two  disjoint 
convex sets can be separated by a hyperplane~\cite{Matousek02}, 
well-separation is equivalent to the fact that for any proper index set $I$, 
the convex hulls of $S_I$ and $S_{[k]\setminus I}$ do not intersect.
A hyperplane $h$ is a \emph{transversal} of $S_1, \dots, S_k$ if
we have $ S_i \cap h \neq \emptyset$, for all $i \in [k]$. 
More generally, for $m \in \{0, \dots, d - 1\}$, an 
\emph{$m$-transversal} of $S_1, \dots, S_k$ is 
an $m$-flat (i.e., an $m$-dimensional affine subspace of $\R^d$) 
that intersects all the $S_i$.
As we shall see below, 
it turns out that well-separation is intimately related to transversals: 
the sets $S_1,\dots, S_k$ are well-separated if and only if there 
is no $(k - 2)$-transversal of the convex hulls of 
$S_1,\dots, S_k$.\footnote{Observe that for any $k \leq d$ sets in $\R^d$, 
there is always a 
$(k - 1)$-transversal: choose one point from each set, 
and consider a $(k - 1)$-flat that goes through these 
points. The $(k - 1)$-flat is unique if the 
chosen points are in general position.}

In the past, transversals have been studied extensively, 
mostly from a combinatorial, but also from a computational perspective.
Arguably the most well-known such theorem is 
\emph{Helly's theorem}~\cite{Helly}, which states that for any 
finite family of convex sets in $\R^d$, it holds that if every $d + 1$ 
of them have a point in common, then all of them do.
In other words, Helly's theorem gives a sufficient \emph{fingerprint} 
condition for a family of convex sets to have a $0$-transversal.
In 1935, Vincensini asked whether such a statement holds 
for general $k$-transversals, that is, whether there is some number 
$m(k, d)$ such that if any $m(k, d)$ sets of a family have a 
$k$-transversal, then all of them do.
This was disproved by Santal\'{o}, who showed that already the 
number $m(1, 2)$ does not exist (cf.~\cite{Handbook} for more details).

One reason why $0$-transversals differ significantly 
from $k$-transversals with $k > 0$ is that the space of $0$-transversals 
of a family of convex sets is itself a convex set. In constrast, 
for $k > 0$, 
the space of $k$-transversals can be very complicated, 
even for pairwise disjoint convex sets.
Thus, in order to generalize Helly's theorem to $k$-transversals with
$k > 0$, additional assumptions become necessary.
For example, Hadwiger's Transversal Theorem~\cite{HadwigerTransversal} 
states that for any family $\mathcal{S}$ 
of compact and convex sets in the plane, it holds that if 
there exists a linear ordering on $\mathcal{S}$ such that any 
three sets can be transversed by a 
directed line in accordance with this ordering, 
then there is a line transversal for $\mathcal{S}$.
This result has been extended to higher dimensions 
by Pollack and Wenger~\cite{PollackWenger}.
Note that to have a well-defined order in which 
a directed line intersects the sets, the
sets should be pairwise disjoint.
Now, \emph{well-separation} is a way to extend this idea 
to transversals of higher 
dimensions: if 
no $k + 1$ sets in a family $\mathcal{S}$ of convex sets 
have a $(k - 1)$-transversal, 
then every $k$-transversal $H$ intersects the set
$\mathcal{S}$ in a well-defined $k$-ordering, that is, for every
way of choosing a $k$-tuple of points from the intersections
of $H$ with $\mathcal{S}$, one point from each set,
the orientation of the resulting simplices is the
same
(that is, they all have the same \emph{order type})~\cite{PollackWenger}.
Under well-separation, the space of transversals becomes simpler, 
in particular for hyperplane transversals: it is now a union of 
contractible sets~\cite{Wenger}.
Note that in $d$ dimensions, there can be no $d + 2$ sets that
are well-separated, due to Radon's theorem which states that 
any set of $d + 2$ points in $d$ dimensions can be partitioned
into two sets whose convex hulls intersect.
For more background on transversals, we refer the interested readers 
to the relevant surveys, e.g.,~\cite{HellySurvey,NewTrends,Handbook}.

Thus, well-separation is a strong assumption on set-families, and 
it does not come as a surprise that for many problems it leads 
to stronger results and faster algorithms compared to the general case.
One such example is obtained for \emph{Ham-Sandwich cuts}, a well-studied
notion that occurs in many places in discrete geometry and topology~\cite{Matousek02}.
Given $d$ point sets $P_1,\dots,P_d$ in $\R^d$, a 
Ham-Sandwich cut is a hyperplane that simultaneously bisects all point sets.
While a Ham-Sandwich cut exists for any family of $d$ point sets~\cite{HS}, 
finding such  a cut is PPA-complete when the dimension is not 
fixed~\cite{GoldbergNecklace}, meaning that it is unlikely that there 
is an algorithm that runs in polynomial time in the dimension $d$.
On the other hand, if $P_1,\dots,P_d$ are well-separated, not only do 
there exist bisecting hyperplanes, but the Ham-Sandwich theorem 
can be generalized to hyperplanes cutting off arbitrary given 
fractions from each point set~\cite{AlphaHSMasses, AlphaHSPoints}.
Further, the problem of finding such a  hyperplane 
lies in the complexity class UEOPL~\cite{ChiuCM20}, a subclass 
of PPA that is believed to be much smaller than PPA.

From an algorithmic perspective, 
the main focus of the previous work have been an efficient
algorithms for finding line transversals 
in two and three dimensions, e.g., see~\cite{LTAgarwal,LTAvis,LTPS}.
To the authors' knowledge, in higher dimensions only algorithms for hyperplane 
transversals have been studied, where the best known algorithm 
for deciding whether a set of $n$ polyhedra with $m$ edges has a 
hyperplane transversal runs in time $O(nm^{d-1})$~\cite{AvisDoskas}.
In particular, there is an exponential dependence on the dimension $d$,
and there are no non-trivial algorithmic results for the
case that the dimension is part of the input.
This curse of dimensionality appears in many geometric problems.
For several problems, it has been shown that there is probably 
no hope to get rid of the exponential dependence in the dimension.
As a relevant example, 
Knauer, Tiwary, and Werner~\cite{knauer2011computational} showed
the following: given $d$ point sets $S_1,\dots,S_d$ in $\R^d$ 
and a point $p\in \R^d$, where $d$ is part of the input, 
it is $W[1]$-hard (and thus NP-hard) to decide whether 
there is there a Ham-sandwich cut for the sets that passes through~$p$.

\subparagraph*{Our Results.}
First, we prove that a family of $k$ sets in $\R^d$ is 
well-separated if and only if the convex hulls of the sets 
have no $(k - 2)$-transversal. This fact seems to be known, but we 
could only find some references without proofs, and some proofs 
of only one direction, for similar definitions of 
well-separation~\cite{castillo2021common,bisztriczky1990separated}. 
Therefore, for the sake of completeness, we present a short proof 
in Section~\ref{sec:wellsepandtrans}. 
This immediately implies that testing well-separation is in coNP.

In \cite{ChiuCM20}, the authors ask 
for the complexity of determining whether a family of point sets 
is well-separated when $d$ is not fixed. 
We present several hardness results for finding $(k-2)$-transversals 
in a family of $k$ sets in $\R^d$. We consider two cases: 
a) finite sets, and b) possibly infinite, but convex set.

\begin{restatable}{theorem}{NPhardTwoPoints}	
	\label{thm:NPhard2points}
	Given a set of $k > d$ point sets in $\R^d$, each with
	at most two points, it is strongly NP-hard to check 
	whether there is a $(d-1)$-transversal, even in the special case $k=d+1$.
\end{restatable}

Note that this decision problem is trivial for $k\leq d$, as the answer 
is always yes. The assumption $k=d+1$ is of special interest 
to us since the transversals we are considering are hyperplanes in $\R^d$, 
as in the Ham-sandwich cuts problem. Moreover, it shows that the problem 
becomes NP-hard for the first non-trivial value of $k$. 
We extend Theorem~\ref{thm:NPhard2points} to show the following:

\begin{restatable}{theorem}{NPhardsegments}
	\label{thm:NPhardsegments}
		Given a set of $k>d$ line segments in $\R^d$, it is NP-hard 
		to check whether there is a $(d-1)$-transversal, 
		even in the special case $k=d+1$.
\end{restatable}

Theorem~\ref{thm:NPhardsegments} implies that testing well-separation 
is coNP-complete even for $d+1$ segments in $\R^d$, answering 
the question from~\cite{ChiuCM20}.
Further, we show the following result, with a stronger hardness 
than Theorem~\ref{thm:NPhard2points}; however, we remove 
the additional constraint that $k=d+1$.

\begin{restatable}{theorem}{StrongNPhardTwoPoints}
	\label{thm:StrongNPhard2points}
	Given a set of $k\leq d+1$ point sets in $\R^d$, each with 
	most two points, it is strongly NP-hard to check whether 
	there is a $(k-2)$-transversal.
\end{restatable}

Observe that for the problem of Theorem~\ref{thm:StrongNPhard2points}, 
we consider $(k-2)$-transversals instead of $(d-1)$-transversals. 
In this context, the problem becomes trivial for $k\geq d+2$, 
because all sets lie in $\R^d$. On the positive side, 
we can show the existence of the following approximation algorithm. 
This can be seen as the special case 
where each point set consists of a single point.

\begin{restatable}{theorem}{approxAlgo}
	\label{thm:approxAlgo}
	Given a set $P$ of $k$ points in $\R^d$, it is 
	possible to compute in polynomial time in $d$ and~$k$ a hyperplane 
	that contains $\Omega(\frac{\textup{OPT} \log k}{k \log \log k}$) points
	of $P$, where $\text{OPT}$ 
	denotes the maximum number of points in $P$ that a hyperplane can contain.
\end{restatable}

In Section~\ref{sec:paramCompl}, we study the problem through 
the lens of parametrized complexity. We show a significant 
difference between finite sets and convex sets.

\begin{restatable}{theorem}{FptAlgo}
	\label{thm:FptAlgo}
	Checking whether a family of $k$ convex sets in $\R^d$ has 
	a $(k-2)$-transversal (or equivalently, whether it is well-separated) 
	is FPT with respect to $d$.
\end{restatable}

\begin{restatable}{theorem}{WHardness}
	\label{thm:W1Hardness}
	Checking whether a family of $k \geq d+1$ finite point sets in $\R^d$ 
	has a $(d-1)$-transversal is W[1]-hard with respect to $d$.
\end{restatable}

Observe that for finite point sets (and, more generally, for any 
non-convex sets), having no $(k-2)$-transversal does 
not a priori imply well-separation. 
The result of Theorem~\ref{thm:W1Hardness} bears a similarity with 
the following result, shown in~\cite{knauer2011computational}: 
given a point set $P$ in $\R^d$, is the origin contained 
in the affine hull of any $d$ points? Indeed, in our reduction in the proof 
of Theorem~\ref{thm:W1Hardness}, one of the point sets contains 
only the origin. However, our proof uses a radically different 
technique, as we have several point sets instead of one, and more 
importantly the number of points one can choose from is $k\leq d+1$, whereas 
in the proof in~\cite{knauer2011computational} the set $P$ 
contains fairly more than $d$ points.

\section{Well-separation and transversals}
\label{sec:wellsepandtrans}
Let us recall some definitions. Let $S_1, \dots, S_k \subset \R^d$ be $k$ sets in $d$ dimensions.
An \emph{$m$-transversal} of $S_1, \dots, S_k$ is an $m$-flat $h \subset \R^d$ (that is, an affine subspace of dimension $m$) with $h \cap S_i \neq \emptyset$ for $i = 1, \dots, k$. 
Transversals are intimately related to well-separation: the sets $S_1, \dots, S_k \subset \R^d$ are well-separated if and only if there is no $(k-2)$-transversal of their convex hulls. As mentioned in the introduction, this fact seems to be well known, but as we could not find a reference with all details for it, we give a short proof for the sake of completeness.
In particular, a $(k-2)$-transversal of the convex hulls is a certificate that $S_1, \dots, S_k$ are not well-separated.
For a given $(k-2)$-flat $h$, it can be checked in polynomial time whether $h$ is a $(k-2)$-transversal, yielding a proof that checking well-separation is in coNP.

\begin{lemma}
Let $S_1, \dots, S_k \subset \R^d$ be $k$ sets in $d$ dimensions.
Then $S_1, \dots, S_k$ are well-separated if and only if their convex hulls have no $(k-2)$-transversal.
\end{lemma}

\begin{proof}
In the following, we assume without loss of generality that the sets $S_1, \dots, S_k$ are convex, that is, they are equal to their convex hulls.
Assume first that $S_1, \dots, S_k$ have a $(k-2)$-transversal~$h$.
Consider the intersection of the sets with $h$.
This gives a collection of $k$ sets $S'_1,\ldots,S'_k$ in a $(k-2)$-dimensional space, thus by Radon's theorem there is an index set $I \subset [k]$ such that the convex hulls of $S'_I$ and of $S'_{[k] \setminus I}$ intersect.
But then also the convex hulls of $S_I$ and of $S_{[k] \setminus I}$ intersect, and thus $S_1, \dots, S_k$ are not well-separated.

For the other direction, assume that $S_1, \dots, S_k$ are not well-separated, that is, there is an index set $I \subset [k]$ such that the convex hulls of $S_I$ and of $S_{[k] \setminus I}$ intersect.
Let $p$ be a point in this intersection.
The point $p$ can be written as a convex combination of points in $S_I$. 
Note that not only can we write it as a convex combination of 
points in $S_I$, but we can even ensure that in 
this combination, we use at most one point of each $S_i$, for $i\in I$. 
This is because the sets $S_i$ are convex and so instead of taking two individual points we can take a convex combination of them. This means that in particular, there is a $(|I|-1)$-transversal $h_I$ of $S_I$ which contains $p$.
The same holds for $S_{[k] \setminus I}$, giving a $(k-|I|-1)$-transversal $h_{[k] \setminus I}$ of $S_{[k] \setminus I}$ which contains $p$.
Then the affine hull of $h_I$ and $h_{[k] \setminus I}$ is a transversal of $S_1, \dots, S_k$ and has dimension at most
$|I|-1+k-|I|-1=k-2$.
\end{proof}

\section{Hyperplane Transversals in High Dimensions}

Let $S_1, \dots, S_k \subset \R^d$ be $k$ sets in $d$ dimensions,
where $d$ is not fixed.
Recall that a \emph{hyperplane transversal} of $S_1, \dots, S_k$ is a $(d-1)$-transversal.
Note that we do not assume the sets to be convex.
In particular, the sets can even be finite.
We consider the decision problem \textsc{HypTrans}: Given sets
$S_1, \dots, S_k$, decide if there is a hyperplane transversal for
them. There are different variants of \textsc{HypTrans}, depending on what
we require from the sets $S_i$. We consider the finite case
and the case of line segments. We also consider the optimisation formulation of \textsc{HypTrans}, that we name \textsc{MaxHyp}:  Given the sets
$S_1, \dots, S_k$, find a hyperplane that intersects as many of these sets as possible.

\subsection{Finite Case}

We begin with the case that all $S_i$ are finite point
sets. We provide an approximation algorithm for \textsc{MaxHyp} in the situation where every $S_i$ contains a single point, for $i = 1, \dots, k$. Note that in this situation, \textsc{HypTrans} can be solved greedily. We also provide some hardness results for \textsc{HypTrans} even in the restricted setting where every $S_i$ contains
at most two points, for $i = 1, \dots, k$.

\subsubsection{Singleton sets}

We assume that every $S_i$ contains a single point, for $i = 1, \dots, k$. We denote by $P$ the point set that is the union of all $S_i$.

\approxAlgo*

\begin{proof}
If $k \leq d$, we output a hyperplane that contains
all points of $P$. 
Otherwise if $k>d$, let $f(k) = \log k / \log \log k$.
If $f(k) < d$, we pick $d$ points from $P$, and
we output a hyperplane through these points.
If $f(k) \geq d$, 
we partition $P$ into disjoint groups of size $f(k)$. 
In each group, we compute all hyperplanes that go through
some $d$ points from the group.
Among all hyperplanes for all groups, we output the hyperplane
that contains the most points in $P$.
For each group, we have 
$O(f(k)^d)=O(f(k)^{f(k)})=O(k)$ hyperplanes to consider. 
Thus, the algorithm runs in polynomial time in $d$ and $k$.

We now analyze the approximation guarantee.
If $f(k) < d$, then we output a
hyperplane with at least $d > f(k) \geq f(k) \text{OPT}/k$ points,
since $\text{OPT} \leq k$.
If $f(k) \geq d$,
we let $h$ be an optimal hyperplane.
If $h$ contains at least $d$ points in a single group,
then we output an optimal solution. Otherwise,
$h$ contains less than $d$ 
points in each group, so
$\text{OPT} \leq d(k/f(k))$. This means that
$d \geq f(k) \text{OPT}/k$, and the claim follows
from the fact that our solution contains at least $d$ points.
\end{proof}

\subsubsection{Sets of at most two points}
\label{subsubsec:twopoints}

Here, we restrict ourselves to the situation  
that every set $S_i$ contains
at most two points, for $i = 1, \dots, k$
We prove that this version of \textsc{HypTrans} is NP-hard, 
with a reduction from \textsc{SubsetSum}.
In \textsc{SubsetSum}, we are given $n + 1$ integers
$a_1, \dots, a_n, b \in \Z$,
and the goal is to
decide whether there exists an index set $I \subseteq \{1, \dots, n\}$
with $\sum_{i \in I} a_i = b$. It is well-known
that \textsc{SubsetSum} is (weakly) NP-complete.
The reduction in this section proofs only a weaker version of Theorem~\ref{thm:NPhard2points}. 
Nevertheless, we think this reduction is shorter and contains the crucial ideas.
In Section~\ref{sec:secondreduction} we then give a second reduction which shows the theorem as stated in the introduction. However the second reduction is more technical.

Given an input $a_1, \dots, a_n, b \in \Z$ for \textsc{SubsetSum},
we create an input $S_1, \dots, S_{n + 2} \subset \R^{n + 1}$
for \textsc{HypTrans}, as follows. Note that the number of sets and the dimension are differing by exactly one.
First, we define $2n + 1$ vectors 
$u, v_1, \dots, v_n, w_1, \dots, w_n \in \R^{n + 1}$, by setting
\begin{align*}
u(1) = -b &\text{ and } u(j) = -1, &&\text{ for } j = 2, \dots, n+1,\\
v_i(1) = a_i &\text{ and } v_i(j) = \delta_{i+1, j}, 
&&\text{ for } j = 2, \dots, n+1, i = 1, \dots, n, \text{ and}\\
w_i(1) = 0 &\text{ and } w_i(j) = \delta_{i+1, j}, 
&&\text{ for } j = 2, \dots, n+1, i = 1, \dots, n.
\end{align*}
Here, for $i, j \in \Z$,
\[
  \delta_{i,j} = \left\{\begin{array}{ll}1, & \text{if } i = j, \\ 0, & 
  \text{if } i \neq j,\end{array}\right.
\]
denotes the \emph{Kronecker delta}.
Using these vectors, we define the input for
\textsc{HypTrans} as $S_1 =  \{v_1, w_1\}, \dots, S_n = \{v_n, w_n\},
S_{n + 1} = \{u\},$ and $S_{n + 2} = \{\mathbf{0}\}$, where 
$\mathbf{0}$ is the origin of $\R^{n + 1}$.

\begin{claim}
We have that $a_1, \dots, a_n, b$ is a yes-input for
\textsc{SubsetSum} if and only if $S_1, \dots, S_{n +2}$ is a
yes-input for \textsc{HypTrans}.
\end{claim}

\begin{proof}
First, suppose that $a_1, \dots, a_n, b$ is a yes-input for \textsc{SubsetSum},
and let $I \subset [n]$ be an index set with $\sum_{i \in I} a_i = b$.
Then, we define a point set $x_1, \dots, x_{n + 2}$ with $x_i \in S_i$ as 
follows: for $i = 1, \dots, n$, if
$i \in I$, we set $x_i = v_i$, and if $i \not\in I$, we set
$x_i = w_i$.
Furthermore, we set $x_{n + 1} = u$ and $x_{n + 2} = \mathbf{0}$.
Then, the points $x_1, \dots, x_{n + 2}$ lie on a common hyperplane of~$\R^{n+1}$.
For this, it suffices to check that
\[
\mathbf{0} = \sum_{i = 1}^{n + 1} \frac{1}{n + 1} x_i,
\]
which follows immediately from the definitions and the choice of the $x_i$.
Thus, there is a hyperplane transversal for $S_1, \dots, S_{n + 2}$, as 
desired.

Conversely, suppose that $S_1, \dots, S_{n + 2}$ is a yes-input
for \textsc{HypTrans}. Thus, there is a choice $x_i \in S_i$, 
for $i = 1, \dots, n + 2$, 
such that $x_1, \dots, x_{n + 2}$,
lie on a common hyperplane. Obviously, it is $x_{n + 1} = u$ and 
$x_{n + 2} = \mathbf{0}$, so we can conclude that $\mathbf{0}$
is in the affine span of $x_1, \dots, x_n, u$ and can be written
as
\[
\mathbf{0} = \sum_{i = 1}^n \lambda_i x_i + \lambda_{n + 1} u,
\]
where $\lambda_i \in \R$ with $\sum_{i = 1}^{n + 1} \lambda_i = 1$.
Let $I \subseteq [n]$ be the set of those indices $i$ for which
$x_i = v_i$.
By inspecting the coordinates and applying the
definitions, we get the following
system of equations.
The first coordinate implies
\begin{align*}
\sum_{i \in I} \lambda_i a_i &= \lambda_{n + 1} b.
\end{align*}
Furthermore for every coordinate $i = 2, \ldots n$, we get 
\begin{align*}
\lambda_i &= \lambda_{n + 1} , &\text{ for } i = 1, \dots, n.
\end{align*}
From the second equation, it follows that $\lambda_1 = \dots = \lambda_{n + 1}$. Since
$\sum_{i = 1}^{n + 1} \lambda_i = 1$, this implies that
$\lambda_i = 1/(n + 1)$, for $i = 1, \dots, n + 1$. Thus, 
the first equation implies that  $a_1, \dots, a_n, b$ is a yes-input
for \textsc{SubsetSum}, with $I$ as the certifying index set.
\end{proof}

\subsubsection{A second reduction}
\label{sec:secondreduction}
Now, we prove that \textsc{HypTrans} is strongly NP-hard even if every set contains at most two points. 
We reduce from \textsc{BinPacking}.
Our reduction will pass through two intermediate problems 
\textsc{EqualBinPacking} and \textsc{FlatTrans}.

We start by defining all the involved problems.

In \textsc{BinPacking}, we are given a sequence
$w_1, \dots, w_n \in  \mathbb{Z}_+$
of \emph{weights},
a number $k$ of \emph{bins} and a 
\emph{capacity} $b\in\mathbb{Z}_+$.
The goal is to decide whether there is a partition of $n$ items
with weights $w_1, \dots, w_n$ into $k$ bins such that in 
each bin the total weight of the items does not exceed the capacity $b$.
It is known that \textsc{BinPacking} is strongly NP-hard.
In \textsc{EqualBinPacking}, we are given the same input, but now 
the goal is to decide whether there exists a partition of the items 
into the bins such that in each bin the total weight of the items 
equals \emph{exactly} the capacity.
Note that \textsc{BinPacking} can easily be reduced to 
\textsc{EqualBinPacking} by adding the appropriate number 
of elements of weight 1, so 
\textsc{EqualBinPacking} is strongly NP-hard as well.

Finally, in \textsc{FlatTrans}, we are given $m$ sets 
$S_1,\dots,S_{m}$ in $\R^{d}$, where $m$ and $d$ are both part of the 
input, and the goal is to decide whether there is an $(m-2)$-transversal. In 
other words, the question is whether there exists an $(m-2)$-dimensional 
affine subspace $h$ such that for all $i \in \{0,\dots,m-1\}$, we 
have $S_i\cap h\neq\emptyset$.
Note that \textsc{HypTrans} with $k=d+1$ is the 
same as \textsc{FlatTrans} with $m=d+1$.

First we show that \textup{\textsc{FlatTrans}} is strongly NP-hard even if all $S_i$ consists of at most two points and $S_m=\{\mathbf{0}\}$.

\StrongNPhardTwoPoints*

\begin{proof}
We reduce from \textsc{EqualBinPacking}.
Given an input $w_1, \dots, w_n, k, b$
to \textsc{EqualBinPacking}, we construct an instance of 
\textsc{FlatTrans} as follows:
we set the dimension $d = k + n + kn$ and the number of sets $m = kn +2$.
For every pair $(i, j) \in [n] \times [k]$, define the vectors
\begin{align*}
&v_{i,j}(x):=\begin{cases} 
      w_i, & \text{if } x = j, \\
      1, & \text{if } x = k + i, \\
      1, & \text{if } x = n + k +  (i - 1)k + j, \\
      0, & \text{otherwise}, 
   \end{cases}
   \text{, } \\
&u_{i,j}(x):=\begin{cases} 
      0, & \text{if } x = j, \\
      0, & \text{if }x = k + i, \\
      1, & \text{if }x = n + k + (i - 1)k +  j, \\
      0, & \text{otherwise}. 
   \end{cases}
\end{align*}
Here, we denote by $x \in \{1, \dots, n+k+kn\}$
the entries of the vector, e.g., 
the first entry of $v_{i,j}$ is denoted by $v_{i,j}(1)$.
Furthermore, let $c$ be the vector whose entries $c(x)$ are 
$-b$, for $1 \leq x \leq k$, and $-1$ everywhere else.
Now set $S_{kn+2}=\{\mathbf{0} \}$, and $S_l=\{v_{i,j},u_{i,j} \}$, 
for $l=(i - 1)k+j$, $i = 1, \dots, n$, $j = 1, \dots, k$ 
(note that this choice of $l$ 
just gives that the order of the $l$'s corresponds 
to the lexicographic order of the $(i,j)$'s) and $S_{kn+1}=\{c\}$.
All these vectors can be constructed in polynomial time.

\begin{claim}
	There is a $kn$-transversal of the sets 
	$S_1,\dots,S_{kn+2}$, if and only if there 
	is a valid solution for the \textsc{EqualBinPacking} instance.
\end{claim}

Assume first that there is a solution for \textsc{EqualBinPacking}.
For each $S_l$, $1\leq l\leq kn$, $l=(i-1)k+j$, 
choose $p_l = v_{i,j}$, if item $i$ is placed in bin $j$, 
and choose $p_l = u_{i,j}$, otherwise. 
Furthermore, set $p_{kn+1} = c$ and $p_{kn+2} = \mathbf{0}$.

We claim that there exist coefficients 
$\lambda_l$ such that 
\begin{equation}\label{equ:sum_zero}
\sum_{l=1}^{kn+1}\lambda_l p_l=\mathbf{0} 
\end{equation}
and 
\begin{equation}\label{equ:convex_comb}
\sum_{l=1}^{kn+1}\lambda_l=kn + 1.
\end{equation}
This implies the claim, because
then $\mathbf{0}$ can be written as a non-trivial linear
combination of the other points. Set $\lambda_l:=1$, for all $l$.
Then, (\ref{equ:convex_comb}) is certainly satisfied.
Consider the $x$'th row of (\ref{equ:sum_zero}), where $1\leq x \leq k$.
By construction, and since we assumed a valid 
solution for the bin packing problem, this row evaluates to
\[
\left(\sum_{i:\text{item $i$ in bin $x$}}w_i\right) - b=0.\]
Similarly, for $k+1\leq x \leq k+n$, the $x$'th row evaluates to
$1-1=0$,
since each item is placed in exactly one bin.
All remaining rows evaluate to $1-1=0$, 
and thus (\ref{equ:convex_comb}) is also satisfied.

Assume now that there exist coefficients $\lambda_l$ that
satisfy (\ref{equ:sum_zero}) and (\ref{equ:convex_comb})
(which must be the case of $\mathbf{0}$ can be written as a non-trivial
linear combination of the other points).
From the $x$'th rows in (\ref{equ:sum_zero}) with $x>k+n$, 
we get $\lambda_l-\lambda_{kn+1}=0$, for $1\leq l \leq kn$, 
and thus $\lambda_1=\dots =\lambda_{kn+1}$.
Together with (\ref{equ:convex_comb}), we thus get 
$\lambda_l=1$, for all $l$.
Put item $i$ into bin $j$ if and only if 
$p_l=v_{i,j}$ for $l=(i-1)k+j$.
Analogous to above we get from the $x$'th rows of 
(\ref{equ:sum_zero}),
for $k+1\leq x \leq k+n$, that each item is placed into exactly one bin.
Further, we get from the $x$'th rows of 
(\ref{equ:sum_zero}), for
 $1\leq x \leq k$, that each bin is filled exactly to capacity.
Thus, we have a valid solution for 
\textsc{EqualBinPacking}, as desired.
\end{proof}
Now, there is only one reduction remaining to show that \textup{\textsc{HypTrans}} is strongly NP-hard even when $S_m=\{\mathbf{0}\}$ and $S_i$ consists of at most two points for all $i = 1, \ldots, m-1$:

\NPhardTwoPoints*

\begin{proof}
We reduce from \textsc{FlatTrans}.
Let us assume that $S_m = \{\mathbf{0}\}$ and let $S_1,\ldots,S_{m}\subset\R^d$ be the sets in the instance of \textsc{FlatTrans}, and assume that $m-1<d$.
We construct sets in $\R^{d+2}$ as follows:
First, for each point $p$ in some set $S_i$ we define the point $p'=(p,0,0)$ and place it in the set $S'_i$.
For $m\leq i\leq d+2$, define $S'_i$ as the set consisting only of the point $s'_i=(0,\ldots,0,1,i)$. Additionally, let $S'_{d+3}:=\{\mathbf{0}\}$.

We claim that $S_1,\ldots,S_{m}\subset\R^d$ have an $(m-2)$-transversal, if and only if $S'_0,S'_1,\ldots,S'_{d+3}\subset\R^{d+2}$ can be transversed by a hyperplane.

Assume first that $S_1,\ldots,S_{m}\subset\R^d$ indeed have an $(m-2)$-transversal, that is, there are points $p_i\in S_i$ and parameters $\lambda_i$ such that $\sum_{i=1}^{m-1}\lambda_i p_i=\mathbf{0}$ and $\sum_{i=1}^{m-1}\lambda_i=1$.
Choosing the corresponding points $p'_i$ and setting $\lambda'_i=\lambda_i$ for $i\leq m-1$ and $\lambda'_i=0$ for $i>m-1$ we get $\sum_{i=1}^{d+2}\lambda'_i p'_i=\mathbf{0}$ and $\sum_{i=1}^{d+2}\lambda'_i=1$, that is, $S'_1,\ldots,S'_{d+3}\subset\R^{d+2}$ can be transversed by a hyperplane.

Assume now that $S'_1,\ldots,S'_{d+3}\subset\R^{d+2}$ can be transversed by a hyperplane, that is, there are points $p'_i\in S'_i$ and parameters $\lambda'_i$, such that $\sum_{i=1}^{d+2}\lambda'_i p'_i=\mathbf{0}$ and $\sum_{i=1}^{d+2}\lambda'_i=1$.
The second to last row of the first equation evaluates to $\sum_{i=m}^{d+2}\lambda'_i=0$, and we thus have $\sum_{i=1}^{m-1}\lambda'_i=1$.
Set $p_i=p'_i$ and $\lambda_i=\lambda'_i$.
Then $\sum_{i=1}^{m-1}\lambda_i=1$ by the observation above.
Further, $\sum_{i=1}^{m-1}\lambda_i p_i=\mathbf{0}$ by the first $m$ rows of the first equation.
Thus, $S_1,\ldots,S_{m}\subset\R^d$ can be transversed by a $(m-2)$-flat.
\end{proof}

\subsection{Line segments}

In this section, we will show that deciding whether there is a hyperplane transversal for $d$ line segments and the origin in $\R^d$, where $d$ is not fixed, is NP-hard.

\NPhardsegments*

We will reduce this to one of the previous cases shown, that is, 
to the restricted version of \textsc{HypTrans} where the sets $S_i$ contain at most two points, see Section \ref{subsubsec:twopoints}.
This is done with the help of a gadget
that enforces that every hyperplane transversal must use
one of the two endpoints of a given line segment. The 
gadget is shown in Figure~\ref{fig:gadget}.
\begin{figure}[h]
	\centering
	\includegraphics{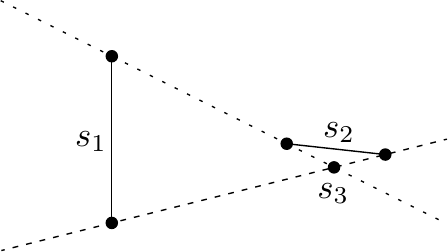}
	\caption{Every hyperplane transversal through $s_1$, $s_2$, $s_3$
	must choose an endpoint of $s_1$ (and of $s_2$).}
	\label{fig:gadget}
\end{figure}

Given a collection of $k$ sets $S_1, \ldots, S_k$ of size at most two, for each set we take the line segment formed by its points as $s_1$, the origin as point $s_3$, and we construct the corresponding new segment $s_2$ using the gadget presented in Figure~\ref{fig:gadget}. This gives a family $S$ of $2k$ line segments that all lie in a $k$-dimensional space.
In order to prove our result, we need to lift our construction to $\R^{2k}$.
Let $A_i$,$B_i$ in $\R^k$ denote the endpoints of the $i$'th original segment ($s_1$ in Figure~\ref{fig:gadget}) and let $G_i$,$H_i$ in $\R^k$ denote the endpoints of the $i$'th gadget segment ($s_2$ in Figure~\ref{fig:gadget}).
Denote by $\varepsilon_j$ the vector in $\R^k$ which is $0$ everywhere except in the $j$'th entry, where it is $\varepsilon$.
Further, we write $\mathbf{0}^k$ for the zero vector in $\R^k$.
We now lift the points $A_i,B_i,G_i,H_i$ to $\R^{2k}$ as follows:
\[A'_i:=\binom{A_i}{\mathbf{0}^k}, B'_i:=\binom{B_i}{\mathbf{0}^k}, G'_i:=\binom{G_i}{\varepsilon_i}, H'_i:=\binom{H_i}{\varepsilon_i}.\]
We denote the corresponding set of line segments $A'_iB'_i$ and $G'_iH'_i$ in $\R^{2k}$ by $S'$.

\begin{claim}
$S\subset\R^k$ has a hyperplane transversal if and only if $S'\subset\R^{2k}$ does.
\end{claim}

\begin{proof}
We will prove this by explicitly computing affine combinations of 
points on the line segments that give us the required transversals.
In this setting, $S\subset\R^k$ has a hyperplane transversal if 
and only if there are real numbers 
$\lambda_i, \gamma_i, \mu_j^{(i)}$, with 
$i \in [k]$, $j \in \{0, \dots, k\}$ 
and the following properties
\begin{equation}\label{eq:origin}
\sum_{i = 1}^k \mu_0^{(i)}\big(\lambda_i A_i + (1 - \lambda_i)B_i\big)
=\mathbf{0}, \\  
\sum_{i = 1}^k \mu_0^{(i)}=1;
\end{equation}
and for all $j \in \{1, \dots, k\}$
\begin{equation}\label{eq:gadgets}
\sum_{i=1}^k \mu_j^{(i)}\big(\lambda_i A_i + (1 - \lambda_i)B_i\big)
=\gamma_jG_j+(1-\gamma_j)H_j,  \\
\sum_{i = 1}^k \mu_j^{(i)}=1.
\end{equation}

Here, the $\lambda_i$ and $\gamma_i$ fix points on the segments, $\mu_0^{(i)}$ write the origin (Equation \eqref{eq:origin}) and the $\mu_j^{(i)}$ the points on the gadget segments (Equation \eqref{eq:gadgets}) as affine combinations of the points on the reduction segments.

Similarly, $S'\subset\R^{2k}$ has a hyperplane transversal if and only if there are real numbers $l_i, g_i, m^{(i)}, n^{(i)}$, with $i\in[k]$ with the following property:
\begin{equation}
\begin{split}
\label{eq:lift}
&\sum_{i = 1}^k m^{(i)}\big(l_i A'_i + (1 - l_i)B'_i\big) + 
\sum_{i = 1}^k n^{(i)}\big(g_i G'_i + (1 - g_i)H'_i\big) = \mathbf{0}, \\
&\sum_{i=1}^k m^{(i)}+n^{(i)}=1.
\end{split}
\end{equation}

Here, the $l_i$ and $g_i$ fix points on the segments and the $m^{(i)}$ and $n^{(i)}$ write the origin as an affine combination of these points.

Assume first that $S\subset\R^k$ has a hyperplane transversal.
Then Equation \eqref{eq:lift} can be satisfied by setting $l_i=\lambda_i, m^{(i)}:=\mu_0^{(i)}, n^{(i)}:=0, g_i:=0$.
Thus, if $S\subset\R^k$ has a hyperplane transversal then so does $S'\subset\R^{2k}$.

As for the other direction, assume that $S'\subset\R^{2k}$ has a hyperplane transversal.
Note that the $(k+i)$'th row of Equation \eqref{eq:lift} reduces to $n^{(i)}\varepsilon=0$, so in particular we must have $n^{(i)}=0$ for every $i\in\{1,\ldots,k\}$.
Thus, we may set $\lambda_i:=l_i$ and $\mu_0^{(i)}:=m^{(i)}$ and Equation \eqref{eq:origin} follows.
As for Equation \eqref{eq:gadgets}, fix some $j\in\{1,\ldots,k\}$ and note that by the construction of the gadget segments there exist real numbers $\alpha_j$ and $\beta_j$ such that $G_j=\alpha_jA_j$ and $H_j=\beta_jB_j$.
Pick real numbers $\gamma_j$ and $x_j$ that satisfy the following two equations:
\begin{equation}
x_j\lambda_j=(1+x_j)\gamma_j\alpha_j,\\ \text{ and } \\
x_j(1-\lambda_j)=(1+x_j)(1-\gamma_j)\beta_j.
\end{equation}
It is straightforward to show that such numbers always exist, for the sake of readability we will not prove this here.
Now, define $\mu_j^{(i)}:=\frac{m^{(i)}}{1+x_j}$ for $j\neq i$ and 
$\mu_j^{(j)}:=\frac{m^{(j)}+x_j}{1+x_j}$.
Then
\[
\sum_{i=1}^k \mu_j^{(i)}\big(\lambda_i A_i+(1-\lambda_i)B_i\big) 
= \frac{1}{1+x_j}\sum_{i=1}^k m^{(i)}\big(l_i A_i+(1-l_i)B_i\big)
+\frac{x_j}{1+x_j}\big(l_i A_i+(1-l_i)B_i\big).
\]
By Equation \ref{eq:lift}, we have 
$\sum_{i=1}^k m^{(i)}\big(l_i A_i+(1-l_i)B_i\big)=0$ 
(recall that $n^{(i)}=0$), thus we have
\[
\sum_{i=1}^k \mu_j^{(i)}\big(\lambda_i A_i+(1-\lambda_i)B_i\big) 
=\frac{1}{1+x_j}\big(x_jl_i A_i+x_j(1-l_i)B_i\big).
\]
From our choice of $\gamma_j$ and $x_j$, we thus get
\[
\frac{1}{1 + x_j} \big(x_jl_i A_i + x_j (1 - l_i)B_i\big)
=\gamma_j\alpha_jA_j+(1-\gamma_j)\beta_j B_j = 
\gamma_j G_j + (1 - \gamma_j) H_j,
\]
which is what we want.
Further, we have
\[\sum_{i=1}^k \mu_j^{(i)}=\frac{1}{1+x_j}\left(\sum_{i=1}^k m^{(i)}+x_j\right)=\frac{1+x_j}{1+x_j}=1,\]
so Equation \eqref{eq:gadgets} is indeed satisfied.
\end{proof}

\section{Parametrized complexity}
\label{sec:paramCompl}

\subsection{An FPT algorithm for \texorpdfstring{$d$}{d} sets}

Recall that our original motivation comes from determining 
whether $d$ sets in $\R^d$ are well-separated. Let us consider 
those $d$ sets, and let us denote by $n$ the total number of 
extreme vertices on their respective convex hulls (for general convex
sets, this might be infninte, but we consider only the finite case). 
We say that $n$ is the \emph{convex hull complexity} of the set family. 
We assume that we are given the extreme points of the convex hull 
of every set and hence have a finite number of points for every set. 

\FptAlgo*

\begin{proof}
For the $O(2^d)$ choices of index sets $I\subset [k]$, we check whether the convex hulls of $S_I$ and $S_{[k]\setminus I}$ intersect. 
For each $I$, we check with an LP whether there is a hyperplane 
separating the points from $S_I$ from the points in $S_{[k]\setminus I}$. 
This can be done by a linear program with $d+1$ variables $
a_0, a_1, \ldots, a_d $ describing a hyperplane in $\R^d$. 
The hyperplane is separating if the following constraints hold:
\begin{equation*}
	a_0 + \sum_{i=1}^d a_i p_i \geq 0 \qquad \text{ for all } p = (p_1, \ldots, p_d) \in S_I \quad \text{ and }
\end{equation*}
\begin{equation*}
	a_0 + \sum_{i=1}^d a_i q_i \leq 0 \qquad \text{ for all } q = (q_1, \ldots, q_d) \in S_{[k] \backslash I }
\end{equation*}
In total we have $O(n)$ constraints. 

If there exists a hyperplane for every $I$, we output that the family is well-separated. 
Thus, there exists a constant $c>0$ such that the total running time of the algorithm is in $O(2^d(nd)^{c}L)$, where $L$ is the number of input bits.

\end{proof}

\subsection{A W[1]-hardness proof}
 We show that \textup{\textsc{FlatTrans}} is $W[1]$-hard with respect to the dimension.
 
\WHardness*

\begin{proof}
We use a framework similar to the one introduced by Marx~\cite{marx2006parameterized}. The reduction is from the following problem: Given a graph $G=(V,E)$ with $n$ vertices, is there a clique of size $k$ in $G$? 

Before describing the point sets, we first explain the framework. We define a set of $k^2$ gadgets, that we call the \emph{encoding gadgets}. To help defining them, we assume that these gadgets lie on $k$ rows and $k$ columns. Note that this representation is purely a help for the definition, but does not correspond to any geometric structure of the point sets we define later. To each gadget we assign a set of admissible tuples $(i,j)$, with $1\leq i,j \leq n$. Let us assume that we are considering the gadget in row $\alpha$ and column $\beta$, with $1 \leq \alpha,\beta \leq k$. If $\alpha=\beta$, the set of admissible tuples is $\{(i,i)\mid 1\leq i \leq n\}$. Otherwise, the set of admissible tuples is $\{(i,j)\mid \{i,j\}\in E\}$. We have in addition the \emph{row gadgets} and the \emph{column gadgets}. A row gadget forces the left value of each encoding gadget from the same row to be the same. Similarly, a column gadget forces the right value of every encoding gadget from the same column to be the same. There is a row gadget for each row, and a column gadget for each column. We say that an encoding is \emph{valid} if each encoding gadget is assigned an admissible tuple, and if all the row and column gadgets are satisfied. As shown by Marx~\cite{marx2006parameterized}, $G$ has a clique of size $k$ if and only if there exists a valid encoding. First let us assume that $v_1,\dots,v_k$ form a clique. Then we assign to the encoding gadget in row $\alpha$ and column $\beta$ the tuple $(v_\alpha,v_\beta)$. Observe that this is an admissible tuple (as there is an edge between $v_\alpha$ and $v_\beta$), and that the encoding is valid since all rows have the same left value, and all columns have the same right value. Reciprocally, let us assume that we have a valid encoding. Assume that the left value of row $\alpha$ is $i$, and that the left value of row $\beta \neq \alpha$ is $j$. Then the encoding gadget in row $\alpha$ and column $\alpha$ is assigned the tuple $(i,i)$, thus column $\alpha$ is assigned right value $i$, which implies that the encoding gadget in row $\beta$ and column $\alpha$ is assigned the tuple $(j,i)$. We have shown that vertices $i$ and $j$ in $G$ are adjacent.

We now describe how to reduce the valid encoding problem to \textsc{FlatTrans}. We define $k^2+2k+2$ point sets in $\R^{k^2+4k}$. Let $k'$ denote $k^2+2k$ and let $k''$ denote $k^2+3k$. We consider the $k'$ gadgets from the framework described above, that is, $k^2$ encoding gadgets as well as $k$ row and $k$ column gadgets, respectively. Let $f$ denote a bijective function from the set of gadgets to $[k']$. For each encoding gadget $g$ in row $\alpha$ and column $\beta$, $1\leq \alpha,\beta \leq k$ we have a point set $P^{\alpha,\beta}$ that contains $O(n^2)$ points. First let us assume $\alpha=\beta$. The point set $P^{\alpha,\alpha}$ contains the points $p^{\alpha,\alpha}_{i}$, for $1\leq i \leq n$, where the coordinates of $p^{\alpha,\alpha}_{i}$ are: $p^{\alpha,\alpha}_{i}(x)=\delta_{f(g),x}+k^i\delta_{k'+\alpha,x}+k^i\delta_{k''+\alpha,x}$. Now let us assume that $\alpha \neq \beta$. The point set $P^{\alpha,\beta}$ contains the points $p^{\alpha,\beta}_{i,j}$, for $1\leq i,j \leq n$ and $\{i,j\} \in E$, where the coordinates of $p^{\alpha,\beta}_{i,j}$ are: $p^{\alpha,\beta}_{i}(x)=\delta_{f(g),x}+k^i\delta_{k'+\alpha,x}+k^j\delta_{k''+\beta,x}$. Now let $g$ be a row gadget, say for row $\alpha$. The point set $P^{\alpha,\cdot}$ contains the points $p^{\alpha,\cdot}_i$, for $1\leq i \leq n$, where $p^{\alpha,\cdot}_i(x)=\delta_{f(g),x}-k^{i+1}\delta_{k'+\alpha,x}$. Similarly, we have a point set $P^{\cdot,\beta}$ for the column gadget $g$ in column $\beta$, and $p^{\cdot,\beta}_i(x)=\delta_{f(g),x}-k^{i+1}\delta_{k''+\beta,x}$ for $1\leq i \leq n$. Finally, we have the point set $P_0=\{\mathbf{0}\}$ and the point set $P_1=\{p_1\}$, where for all $1\leq x \leq k'$, $p_1(x)=-1$, and $p_1(x)=0$ otherwise. Observe that we have indeed $k^2+2k+2$ point sets of size $O(n^2)$ in $\R^{k^2+4k}$. The absolute values of all point coordinates are at most $k^{n+1}$. Thus, we can describe it with $\log (k^{n+1})=(n+1)\log(k)$ bits. We claim that there is a $(k^2+2k)$-transversal if and only if $G$ has a clique of size $k$. From the reduction, this immediately implies that \textsc{FlatTrans} is $W[1]$-hard with respect to the dimension.

First let us assume that there is a clique of size $k$ in $G$. From what we argued, it implies that there is a valid encoding of the gadgets. We define a set of $k'+1$ points as follows. First we take the point $p_1$. If the tuple assigned to gadget in row $\alpha$ and column $\beta\neq \alpha$ is $(i,j)$, then we take the point $p^{\alpha,\beta}_{i,j}$. If the gadget in row $\alpha$ and column $\alpha$ is assigned the tuple $(i,i)$, then we take the point $p^{\alpha,\alpha}_i$. Likewise, if the left value of row $\alpha$ is $i$, we take the point $p^{\alpha,\cdot}_{i}$. Finally, if the right value of column $\beta$ is $j$, we take the point~$p^{\cdot,\beta}_j$. We denote those $k'+1$ points by $p_1,\dots,p_{k'+1}$ and claim that they lie on a common hyperplane which contains $\mathbf{0}$. It suffices to show that 
\begin{equation*}
\sum_{1\leq \ell\leq k'+1}\frac{1}{k'+1}p_\ell=\mathbf{0}.
\end{equation*}
Consider the first $k'$ coordinates. Recall that $f$ is a bijection between the set of gadgets and $[k']$ and recall that by definition, the points $p_\ell$ have exactly one entry $1$ in the first $k'$ coordinates. Therefore in this sum, we have exactly one entry $1$ from exactly one of the gadgets and exactly one entry $-1$ from the point $p_1$ in each of these coordinates. So it is clear that this equation is satisfied in the first $k'$ coordinates.
Now let us consider the coordinate $k'+\alpha$, for some $1\leq \alpha \leq k$. As the encoding is valid, it implies that the left value in row $\alpha$ of all encoding gadgets is the same. Let us denote by $i$ this left value. We have indeed 
\[
\sum_{1\leq \ell\leq k'+1}\frac{1}{k'+1}p_\ell(k'+\alpha)
=\frac{1}{k'+1}\left(\left(\sum_{1\leq \beta \leq k}k^i\right)-k^{i+1}
\right)=0.
\]
Likewise if the coordinate is of the form $k''+\beta$ for some $1\leq \beta \leq k$, we argue using the fact that the right value of all encoding gadgets in column $\beta$ is the same. This completes the first direction of our proof.

For the second direction, let us assume that there is a hyperplane $h$ that contains at least one point from each point set. By assumption one of these points is $\mathbf{0}$, another is $p_1$, and we denote the others by $p_2,\dots,p_{k'+1}$. This implies that we have $\mathbf{0}=\lambda_1p_1+\sum_{2\leq \ell\leq k'+1}\lambda_\ell p_\ell$, where $\lambda_\ell \in \R$ and $\sum_{1\leq \ell \leq k'+1}\lambda_\ell=1$. By looking at the $k'$ first coordinates, we immediately obtain $\lambda_1=\lambda_i=\frac{1}{k'+1}$, for all $2\leq i \leq k'+1$. Let assume that in point set $P^{\alpha,\beta}$ with $1\leq \alpha,\beta \leq k$, the point $p^{\alpha,\beta}_{i,j}$ is contained in $h$, for some $1\leq i,j \leq n$. Note that by definition, $(i,j)$ is an admissible tuple of the encoding gadget in row $\alpha$ and column $\beta$. We assign this tuple to this gadget, and do likewise with all other encoding gadgets. It remains to show that the left value of all encoding gadgets in the same row is the same, and that the same holds with the right value of encoding gadgets from the same column. Let us consider row~$\alpha$. We consider the points contained in $h$ that belong to $P^{\alpha,\beta}$, for some $1\leq \beta \leq k$. Let us denote by $Y$ the set of their $(k'+\alpha)$-th coordinate. Let $z$ be equal to $\max \{\log_k(y)\mid y \in Y\}$. By assumption, we know that $\sum_{y\in Y}y=k^i$ for some $2\leq i \leq n+1$. This is because the coefficients $\lambda_\ell$ for these point sets are equal to the coefficient for the point in $P^{\alpha,\cdot}$ contained in $h$. As the elements in $Y$ are non-negative, we obtain $i\geq z+1$. Assume for a contradiction that not all elements in $Y$ are equal. Then we have $\sum_{y\in Y}y<\sum_{y\in Y}k^z=k^{z+1}\leq k^i$. As this is not possible, we know that all elements in $Y$ are equal, which implies that the left value of all encoding gadgets in row $\alpha$ is the same. We can argue likewise for the columns. 
\end{proof}

\section{Conclusion and Open Problems}
We showed that the problem of testing well-separability of $k$ sets in $\R^d$ is hard. However, it may be that there exist some algorithms which solve the problem in a smarter way than simply testing the $2^k$ choices of index set. This question is still open.

It would be interesting to have some inapproximability results, or some better approximation algorithms, for the problem of finding a hyperplane that intersects as many points as possible in a point set $P$ in $\R^d$, where $d$ is not fixed.

\bibliographystyle{plain}
\bibliography{literature}

\newcommand{\SortNoop}[1]{}
\begin{thebibliography}{10}

\bibitem{LTAgarwal}
Pankaj~K. Agarwal.
\newblock On stabbling lines for convex polyhedra in 3d.
\newblock {\em Comput. Geom. Theory Appl.}, 4(4):177--189, 1994.

\bibitem{HellySurvey}
Nina Amenta, Jes{\'u}s~A {De Loera}, and Pablo Sober{\'o}n.
\newblock Helly's theorem: new variations and applications.
\newblock {\em arXiv preprint arXiv:1508.07606}, 2015.

\bibitem{AvisDoskas}
David Avis and Mike Doskas.
\newblock Algorithms for high dimensional stabbing problems.
\newblock {\em Discrete applied mathematics}, 27(1-2):39--48, 1990.

\bibitem{LTAvis}
David Avis and Rephael Wenger.
\newblock Algorithms for line transversals in space.
\newblock In {\em Proc. 3rd Annu. Sympos. Comput. Geom. (SoCG)}, pages
  300--307, 1987.

\bibitem{AlphaHSMasses}
Imre B{\'a}r{\'a}ny, Alfredo Hubard, and Jes{\'u}s Jer{\'o}nimo.
\newblock Slicing convex sets and measures by a hyperplane.
\newblock {\em Discrete Comput. Geom.}, 39(1-3):67--75, 2008.

\bibitem{bisztriczky1990separated}
Ted Bisztriczky.
\newblock On separated families of convex bodies.
\newblock {\em Archiv der Mathematik}, 54(2):193--199, 1990.

\bibitem{castillo2021common}
Federico Castillo, Joseph Doolittle, and Jose~Alejandro Samper.
\newblock Common tangents to polytopes.
\newblock {\em arXiv preprint arXiv:2108.13569}, 2021.

\bibitem{ChiuCM20}
Man-Kwun Chiu, Aruni Choudhary, and Wolfgang Mulzer.
\newblock Computational complexity of the {$\alpha$}-ham-sandwich problem.
\newblock In {\em Proc. 47th Internat. Colloq. Automata Lang. Program.
  (ICALP)}, pages 31:1--31:18, 2020.

\bibitem{GoldbergNecklace}
Aris Filos-Ratsikas and Paul~W. Goldberg.
\newblock The complexity of splitting necklaces and bisecting ham sandwiches.
\newblock In {\em Proceedings of the 51st Annual ACM SIGACT Symposium on Theory
  of Computing}, pages 638--649, 2019.

\bibitem{NewTrends}
Jacob~E. Goodman, Richard Pollack, and Rephael Wenger.
\newblock Geometric transversal theory.
\newblock In {\em New trends in discrete and computational geometry}, pages
  163--198. Springer, 1993.

\bibitem{HadwigerTransversal}
Hugo Hadwiger.
\newblock {\"U}ber {E}ibereiche mit gemeinsamer {T}reffgeraden.
\newblock {\em Portugaliae mathematica}, 16(1):23--29, 1957.

\bibitem{Helly}
Eduard Helly.
\newblock {\"U}ber {M}engen konvexer {K}{\"o}rper mit gemeinschaftlichen
  {P}unkten.
\newblock {\em Jahresbericht der Deutschen Mathematiker-Vereinigung},
  32:175--176, 1923.

\bibitem{Handbook}
Andreas Holmsen and Rephael Wenger.
\newblock 4 {H}elly-type theorems and geometric transversals.
\newblock {\em Handbook of Discrete and Computational Geometry}, 2017.

\bibitem{knauer2011computational}
Christian Knauer, Hans~Raj Tiwary, and Daniel Werner.
\newblock On the computational complexity of ham-sandwich cuts, helly sets, and
  related problems.
\newblock In {\em Proc. 28th Sympos. Theoret. Aspects Comput. Sci. (STACS)},
  volume~9, pages 649--660, 2011.

\bibitem{marx2006parameterized}
D{\'a}niel Marx.
\newblock Parameterized complexity of independence and domination on geometric
  graphs.
\newblock In {\em International Workshop on Parameterized and Exact
  Computation}, pages 154--165. Springer, 2006.

\bibitem{Matousek02}
Ji\v{r}\'{\i} Matou\v{s}ek.
\newblock {\em Lectures on discrete geometry}, volume 212 of {\em Graduate
  Texts in Mathematics}.
\newblock Springer-Verlag, New York, 2002.

\bibitem{LTPS}
Marco Pellegrini and Peter~W. Shor.
\newblock Finding stabbing lines in 3-space.
\newblock {\em Discrete Comput. Geom.}, 8(2):191--208, 1992.

\bibitem{PollackWenger}
Richard Pollack and Rephael Wenger.
\newblock Necessary and sufficient conditions for hyperplane transversals.
\newblock {\em Combinatorica}, 10(3):307--311, 1990.

\bibitem{AlphaHSPoints}
William Steiger and Jihui Zhao.
\newblock Generalized ham-sandwich cuts.
\newblock {\em Discrete Comput. Geom.}, 44(3):535--545, 2010.

\bibitem{HS}
Arthur~H. Stone and John~W. Tukey.
\newblock Generalized ``sandwich'' theorems.
\newblock {\em Duke Math. J.}, 9(2):356--359, 06 1942.

\bibitem{Wenger}
Rephael Wenger.
\newblock {\em Geometric permutations and connected components}.
\newblock DIMACS, Center for Discrete Mathematics and Theoretical Computer
  Science, 1990.

\end{thebibliography}
\end{document}